\documentclass[reqno,12pt,letterpaper]{amsart}
\usepackage{amsmath,amssymb,amsthm,graphicx,mathrsfs,url}
\usepackage[usenames,dvipsnames]{color}
\usepackage[colorlinks=true,linkcolor=Red,citecolor=Green]{hyperref}
\usepackage{amsxtra}
\usepackage{subcaption}
\usepackage{float}
\usepackage{hyperref}
\usepackage{mathabx}

\def\arXiv#1{\href{http://arxiv.org/abs/#1}{arXiv:#1}}

\def\?[#1]{\textbf{[#1]}\marginpar{\Large{\textbf{??}}}}
\def\smallsection#1{\smallskip\noindent\textbf{#1}.}
\let\epsilon=\varepsilon 

\newcommand{\RR}{{\mathbb R}}

\newcommand{\CC}{{\mathbb C}}

\newtheorem{theo}{Theorem}
\newtheorem{prop}{Proposition}[section]	
\newtheorem{defi}[prop]{Definition}

\newtheorem{ass}{Assumption}
\newtheorem{ex}{Example}

\numberwithin{equation}{section}

\DeclareMathOperator{\Spec}{Spec}

\DeclareMathOperator{\tr}{tr}

\def\indic{\operatorname{1\hskip-2.75pt\relax l}}
\usepackage{scalerel}

\newcommand\reallywidehat[1]{\arraycolsep=0pt\relax%
\begin{array}{c}
\stretchto{
  \scaleto{
    \scalerel*[\widthof{\ensuremath{#1}}]{\kern-.5pt\bigwedge\kern-.5pt}
    {\rule[-\textheight/2]{1ex}{\textheight}} 
  }{\textheight} %
}{0.5ex}\\           
#1\\                 
\rule{-1ex}{0ex}
\end{array}
}

\title[Quantum statistical learning]{Quantum statistical learning via Quantum Wasserstein natural gradient}
\author{Simon Becker}
\email{simon.becker@damtp.cam.ac.uk}
\address{DAMTP, University of Cambridge, Wilberforce Road, Cambridge CB3 0WA, UK}
\author{Wuchen Li}
\email{wuchen@mailbox.sc.edu}
\address{Department of Mathematics, University of South Carolina, USA}

\begin{document}

\begin{abstract}
In this article, we introduce a new approach towards the statistical learning problem $\operatorname{argmin}_{\rho(\theta) \in \mathcal P_{\theta}} W_{Q}^2 (\rho_{\star},\rho(\theta))$ to approximate a target quantum state $\rho_{\star}$ by a set of parametrized quantum states $\rho(\theta)$ in a quantum $L^2$-Wasserstein metric. We solve this estimation problem by considering Wasserstein natural gradient flows for density operators on finite-dimensional $C^*$ algebras. For continuous parametric models of density operators, we pull back the quantum Wasserstein metric such that the parameter space becomes a Riemannian manifold with quantum Wasserstein information matrix. Using a quantum analogue of the Benamou-Brenier formula, we derive a natural gradient flow on the parameter space. We also discuss certain continuous-variable quantum states by studying the transport of the associated Wigner probability distributions. 
\end{abstract}
\maketitle

\section{Introduction}
The learning problem of quantum states, i.e. positive-definite trace class operators of unit trace, is central in modern quantum theory and commonly called \emph{quantum state tomography}. The problem of quantum state estimation is ubiquitous in quantum mechanics and has a wide range of applications: This includes the analysis of optical devices \cite{DLPPS02} as well as the reliable estimation of qubit states in quantum computing \cite{BK10,LR09}. Until this day, there have been many recent computationally efficient approaches towards the quantum state estimation problem based on compressed sensing and machine learning methods such as \cite{GLFBE,TMC18}. For a review of the most common classical approaches towards quantum state estimation, such as Maximal likelihood estimation (MLE), we refer to \cite{PR04}.

\medskip 
However, both in physics and non-commutative geometry, many problems come as a quantum state estimation problem in disguise:
Over the past years, finding suitable physical descriptors for molecular structures from data has become a vast and growing area of research, cf. the review article \cite{SMB19} and references therein. Recently, such quantum machine learning approaches have also been based on optimization problems in Wasserstein distances, see for example \cite{CLB20}, where a kernel ridge regression-based model relying on the Coulomb matrix is studied. The advantage of using the Wasserstein distance is that it leads to a continuous dependence on the position of the nuclei. 

In said article, it has been discovered that it is key to use a suitable parametrization of the Coulomb matrix. This parametrization is ought to be invariant under 3D translations and rotations of the molecule and therefore related to the low-dimensional parametrization problem considered in this and previous articles, cf. \cite{CL18}. Also, first attempts towards quantum Wasserstein generative adversarial networks have been considered in \cite{CHLFW19}. The quantum Wasserstein distance and its generalizations considered in \cite{CGT,CGT2} have also far-reaching applications beyond quantum mechanics to the field of non-commutative probability theory which includes multivariable time series and vector-valued random variables \cite{NGT}. Hence, solving the quantum state estimation problem in Wasserstein distance has become an important and widely applicable problem.

The analysis of geometric properties of the space of quantum states is called quantum information geometry and is central in the field of quantum information. The asymptotic theory of quantum state estimation and quantum information geometry has been developed in the second half of the 1980s by Nagaoka \cite{N95}. A comprehensive review of the modern field of quantum information geometry and its connection to quantum estimation can be found in \cite{H17}. In this article, we develop a new connection between these two fields based on the quantum Wasserstein metric. 

It has been discovered, among others, by Otto \cite{O01}, that various PDEs evolve according to the gradient flow with respect to the $L^2$-Wasserstein metric \cite{L1998}. Later, Carlen and Maas introduced in a series of articles \cite{CM14,CM18,CM20} also quantum Wasserstein metrics for open quantum systems, satisfying a detailed balance condition. In these articles, they showed, that such open quantum systems also evolve according to the $L^2$-Wasserstein gradient flow. Moreover, they also showed that their metric allows for a dynamical formulation extending the classical Benamou-Brenier formula \cite{BB} to the quantum setting. Here, we also mention the work by Datta and Rouz\'e \cite{RD,RD2} for additional links to a quantum version of Ricci curvature and Fisher information functional. This analysis has been complemented by articles \cite{CGT,CGT2} where different types of non-commutative multiplication operators are considered with favorable properties from a computational point of view. Besides, Carlen and Maas showed that for certain open quantum systems the gradient flow of the relative entropy with respect to an invariant state in the quantum Wasserstein metric coincides with the quantum evolution governed by the Lindblad equation. For continuous-variable states, a quantum transport framework with desirable physical features has been proposed in \cite{DPT19}. However, a dynamical formulation of this approach does not seem to exist, yet. Results on the entropy flow for open quantum systems have also been obtained in \cite{MM17}. Another relevant definition of the Wasserstein distance is due to Golse, Mouhot, and Paul \cite{GMP} and has been proposed in the study of uniform mean-field limits of quantum systems in the semiclassical parameter.

Recently, optimal transport gradient flows have been applied to estimation problems in classical probability theory. In particular, the parameter estimation problem of probability measures by using parameterized Wasserstein gradient flows on either Kullback-Leibler (KL) divergence, also referred to as relative entropy, or $L^2$-Wasserstein distance has been addressed by the second author \cite{CL18,LM18,LM20}. This leads to a joint study between optimal transport \cite{OT} and information geometry \cite{IG, IG2}, namely transport information geometry \cite{Li1,Li2}. Here, the natural gradient induced by optimal transport is first applied for statistical learning problems. Meanwhile, this approach also introduces a new estimation theory based on Wasserstein information matrix \cite{LZ20}. It also develops new scientific computing algorithms by the generative adversarial network to solve classical Fokker-Planck equations, in data-poor situations \cite{LLZZ19}.

In this article, we present a new approach towards quantum state estimation based on $L^2$-quantum Wasserstein gradients. We extend the study of the previous paragraph to quantum systems. We start by studying the problem of minimizing the distance with respect to a quantum Wasserstein metric $d$, for some fixed target density operator $\rho_{\star}$ over a parametrized manifold of states  $\mathcal P_{\theta} \subset \mathscr D(\mathcal H)$, i.e. we aim to identify $ \operatorname{argmin}_{\rho \in \mathcal P_{\theta}} d(\rho_{\star},\rho).$ We address the corresponding estimation problem for particular finite and infinite-dimensional quantum states. In the case of infinite-dimensional states, our approach towards statistical learning is based on the Wigner transform of continuous-variable quantum states. This makes this approach particularly tailored to experimental quantum state estimation in continuous-variable systems, where the Wigner distribution of the quantum state is approximately recovered \cite{VR89}. A classical choice of the distance between probability measures is the Kullback-Leibler (KL) divergence. In classical probability, the metric induced by the $L^2$ Hessian of the KL divergence is the \emph{Fisher-Rao metric} which provides a natural gradient descent method. The analogous concepts of relative entropy for faithful states $\rho$ and $\sigma$
\[ S(\rho \Vert \sigma) = - \operatorname{tr}(\rho(\log(\rho)-\log(\sigma))) \]
and Fisher information is well-established in quantum information theory, too. For finite-dimensional quantum states, our aim is then to establish low-dimensional parameterized quantum Wasserstein gradient flows based on quantum Wasserstein distances. This means we aim to find a low-dimensional representation of the minimization problem in parameter space by applying quantum Wasserstein dynamics. Our study starts by pulling back the quantum Wasserstein metric to a finite-dimensional parameter manifold, using the quantum transport (Wasserstein) information matrix. This leads to a natural gradient descent method for quantum states.

We also introduce and study a quantum analog of the Schr\"odinger bridge problem. As we show in this article, this problem can be solved by a quantum Benamou-Brenier's formula with quantum Fisher information functional regularization. 

\smallsection{Summary of novel results:}
\begin{itemize}
\item We introduce the quantum transport information matrix and develop the related quantum transport/Wasserstein statistical manifold. 
This can be viewed as the first step of quantum transport information geometry. 
\item We formulate the quantum transport natural gradient flow based on quantum Wasserstein statistical manifold. We apply this flow for solving the quantum statistical learning problem.  
\item We also formulate the quantum Schr{\"o}dinger bridge problem by controlling the quantum transport natural gradient flows. 
\item We study the quantum Wasserstein statistical manifold for various finite-dimensional systems such as the quantum fermionic Fokker-Planck dynamics and more general finite-dimensional open quantum systems satisfying the detailed balance condition, as well as for continuous-variable systems with positive Wigner functions such as (mixtures of) Gaussian states.
\item We illustrate our results on some simple examples and also discuss how they apply to the parameter estimation problem for quantum channels.
\end{itemize}
\smallsection{Outline of the article}
In Section \ref{sec:Review} we provide a brief review of classical optimal transport theory and quantum optimal transport, i.e.  
\begin{itemize}
\item Classical optimal transport, Sec.\ref{sec:Review1}.
\item Natural gradient flow, Sec.\ref{sec:NGF}.
\item Schr\"odinger bridge problem, Sec.\ref{sec:clSBP}.
\item Quantum optimal transport, Sec.\ref{sec:QOT}.
\item Fermionic Fokker-Planck equation, Sec.\ref{subsec:FFP}.
\item Quantum Markov semigroups satisfying detailed balance (DB), Sec.\ref{sec:QMSDB}.
\end{itemize}
In Section \ref{sec:DBC} we then introduce the quantum Wasserstein natural gradient \ref{eq:GFFF}, the Schr\"odinger bridge problem for finite-dimensional quantum systems in Section \ref{sec:QSBP}, and the same two for certain continuous-variable systems, including Gaussian systems, in Section \ref{sec:CVS}. In Section \ref{sec:Ex} we discuss examples of our theory. This includes the transport problem for two Gaussian states and a fully explicit case of the fermionic Fokker-Planck equation. We finish our collection of examples by illustrating how the quantum transport information matrix can also be used to perform parameter estimation for quantum channels. 

\smallsection{Notation}
We denote by states $\vert n \rangle$, for $n \in \mathbb N_0$,  the canonical eigenbasis of the number operator $N= a^*a$ where $a$ is the standard annihilation operator. 
The continuous linear operators on a normed space $X$ are denoted by $\mathcal L(X)$, the space of trace-class operators on a Hilbert space $\mathcal H$ by $\operatorname{TC}(\mathcal H).$ For a set $\Omega$ we denote by $\operatorname{int}(\Omega)$ its interior. The set of quantum states (positive-definite operators of unit trace) on a Hilbert space $\mathcal H$ is denoted by $\mathscr D(\mathcal H).$  We denote the Riemannian manifold of faithful states by $\mathscr D_+(\mathcal H).$ We recall that $\partial \mathscr D(\mathcal H)$ are states with zero determinant and $\operatorname{int}(\mathscr D(\mathcal H)) = \mathscr D_+(\mathcal H).$ We also write $\{X, Y\}=XY+YX$ for the anti-commutator and $[X,Y]=XY-YX$ for the commutator. We denote the spectrum of a linear operator $T$ by $\Spec(T).$

\section{Review of classical and quantum optimal transport}
\label{sec:Review}

Our goal is to study the problem of minimizing the distance with respect to a $L^2$-quantum Wasserstein distance $ W^2_Q$, for some fixed target density operator $\rho_{\star}$ over a parametrized manifold of states  $\mathcal P_{\theta} \subset \mathscr D(\mathcal H)$, i.e. we aim to identify $ \operatorname{argmin}_{\rho \in \mathcal P_{\theta}} W^2_Q(\rho_{\star},\rho)$.

For this purpose, we start in this section with a review of the classical framework and highlight similarities and differences that appear in the quantum setting. In addition, we will also employ the classical framework for the study of Wigner distributions in the continuous-variable setting.

\subsection{Classical optimal transport}
\label{sec:Review1}

The optimal transport problem dates back to 1781 when Monge asked how to find for two probability measures $f_0,f_1$ on $\Omega \subset \RR^n$, with finite second moment, an optimal transport plan $T:\Omega \rightarrow \Omega$ pushing $f_0$ to $f_1$ such that the transportation cost is minimized and for all $A \subset \Omega$ measurable
\[ \inf_T \int_{\Omega} \Vert x- T(x) \Vert^2 f_0(x) \ dx : T_{*}f_0  = f_1\]

For two probability measures with densities $f_0,f_1$ on $\Omega \subset \RR^n$ the square of the classical $L^2$-Wasserstein distance is defined as 
\begin{equation}
\label{eq:L2Wasserstein}
W^2_{\text{cl}}(f_0,f_1):=\inf_{\pi \in \Pi(f_0,f_1)} \int_{\Omega \times \Omega} \vert x-y \vert^2 d\pi(x,y) 
\end{equation}
where $\Pi(f_0,f_1)$ is the set of all couplings of the two measures $f_0(x) \ dx$ and $f_1(x) \ dx.$

Equivalent to \eqref{eq:L2Wasserstein}, and particularly relevant for our purposes, is a dynamical formulation, given by the Benamou-Brenier formula, which states that the Wasserstein metric is given by
\begin{equation}
\label{eq:BBF}
 W^2_{\text{cl}}(f_0,f_1)=\inf \int_0^1 \int_{\Omega} \vert v_t(x) \vert^2 \ d\mu_t(x)  dt 
 \end{equation}
where the infimum is taken over all pairs $(\mu_t,v_t)$ where $\mu_t$  with $\mu_0=f_0 $ and $\mu_1 = f_1 $ is a curve of measures and $v_t$ a time-dependent vector field satisfying 
\[ \partial_t \mu_t + \operatorname{div}(v_t \mu_t) = 0.\]
On a bounded domain $\Omega$ the above formulation is replaced by the corresponding Neumann problem. 

The dynamical formulation above is closely connected to a Riemannian structure on the Wasserstein space. To fix ideas, we consider the space of strictly positive densities 
\[ \mathscr D_{+}(\Omega)=\{ f \in C^{\infty}(\Omega, (0,\infty)): \Vert f \Vert_{L^1}=1\}.\]

The tangent space of $\mathscr D_{+}$ is then just given by
\[ T_{f}\mathscr D_{+}(\Omega)=\left\{ \sigma \in C^{\infty}(\Omega): \int_{\Omega} \sigma(x) \ dx =0\right\}.\]

For any $\Phi \in C^{\infty}(\Omega)$ we can then set
\[ V_{\Phi}(x) := - \operatorname{div}(f(x) \nabla \Phi(x)) \in T_{f}\mathscr D_{+}(\Omega).\]

This map provides an isomorphism, at least if $\Omega$ is compact, 
\[C^{\infty}(\Omega)/ \RR \rightarrow T_{f}\mathscr D_{+}(\Omega), \text{ with }[\Phi] \mapsto V_{\Phi}.\]

We can therefore define the $L^2$-Wasserstein metric tensor by introducing:
\begin{defi}[$L^2$-Wasserstein metric tensor] We define the metric tensor $g_{f}:T_{f}\mathscr D_{+}(\Omega) \times T_{f}\mathscr D_{+}(\Omega) \rightarrow \RR$ by 
\begin{equation}
\label{eq:metrictens}
 g_{f}(\sigma_1,\sigma_2) :=\int_{\Omega} \langle \nabla \Phi_1(x),\nabla \Phi_2(x) \rangle f(x) \ dx, 
 \end{equation}
with $\sigma_i = V_{\Phi_i}.$
\end{defi}

\subsection{Natural gradient flow}
\label{sec:NGF}
We continue with a review of the main results of \cite[Sec. $3$]{CL18} and explain how to minimize an objective function efficiently in parameter space.

We define the statistical parameter space as a $d$-dimensional Riemannian manifold $\Theta$ with connection $D_{\theta}$ and metric tensor $\langle \xi, \eta \rangle_{\theta} = \xi^T G_{\theta} \eta$. We then take a continuous parametrization $\Theta \ni \theta \mapsto \rho(\bullet,\theta) \in \mathscr D_{+}(\Omega)$ and introduce a natural metric tensor by pulling back \eqref{eq:metrictens} on the statistical manifold
\[ g_{\theta}: T_{\theta}(\Theta)^2 \rightarrow \RR, \text{ such that } g_{\theta}(\xi, \eta) = g_{\rho(\bullet, \theta)}(D_{\theta}\rho(\xi),D_{\theta}\rho(\eta))= \langle \xi,G_W(\theta) \eta \rangle \]
where $G_W(\theta) = \left( G_{\theta}^* \  \int_{\Omega} \partial_{\theta_i} \rho(x,\theta) (-\operatorname{div}(\rho(x,\theta)\nabla))^{-1} \partial_{\theta_j}\rho(x,\theta) \ dx \  G_{\theta} \right)_{ij}.$

The Wasserstein natural gradient is then for an objective function $R(\theta)$ defined by 
\[ \dot{\theta}(t) = -\nabla_{g} R(\theta(t))\]
where $\nabla_g$ is the unique gradient vector satisfying 
\[ g_{\theta}( \nabla_g R(\theta),\xi)= \langle D_{\theta} R(\theta), \xi \rangle_{\theta}.\]
In particular, we have the identification $\nabla_g R(\theta) = G_W(\theta)^{-1} G_{\theta} D_{\theta} R(\theta).$

The \emph{Wasserstein gradient descent} can then be numerically implemented using a standard forward Euler method 
\[ \theta_{(n+1)\tau} := \theta_{n\tau}- \tau G_W(\theta_{n\tau})^{-1} G_{\theta} D_{\theta} R(\rho(\bullet, \theta_{n \tau})).\]
This gradient flow method can be interpreted as an approximate solution to the minimization problem
\[ \operatorname{argmin}_{\theta \in \Theta} R(\rho(\bullet, \theta)) + \frac{W^2_{\text{cl}}( \rho(\bullet, \theta_{n \tau}),\rho(\bullet, \theta))^2 }{2 \tau}\]
which is obvious from considering the linearized expressions
\begin{equation}
\begin{split}
W^2_{\text{cl}}( \rho(\bullet, \theta+\Delta \theta),\rho(\bullet, \theta))^2  &= \frac{1}{2} \langle \Delta \theta, G_W(\theta) \Delta \theta \rangle + o((\Delta \theta)^2)  \\
R(\rho(\bullet, \theta+\Delta \theta)) &= R(\rho(\bullet, \theta)) + \langle D_{\theta} R(\rho(\bullet, \theta)), \Delta \theta \rangle_{\theta} + o(\Delta \theta ).
\end{split}
\end{equation}

\subsection{Fisher information regularization and Schr\"odinger bridge problem}
\label{sec:clSBP}
After the works of Monge and Kantorovich, Schr\"odinger proposed in 1931 a similar transport problem which is nowadays referred to as the \emph{Schr\"odinger bridge problem} (SBP):

Given two strictly positive densities $f_0,f_1$ on a domain $\Omega \subset \RR^n$, consider 

\begin{equation}
\label{eq:functional}
\inf_{m,\rho} \int_0^1 \int_{\Omega} \frac{m(t,x)^2}{f(t,x)} \ dx \ dt, 
\end{equation}
where the infimum is taken over all $m$ and $f$ satisfying 
\begin{equation}
\label{eq:FP}
\partial_t f(t,x) + \operatorname{div}(m(t,x)) = \beta \Delta f(t,x), \quad f(0,x)=f_0(x), f(1,x)=f_1(x) 
\end{equation}
with the boundary condition
\[ \langle m(t,x) -\nabla f(t,x), n(x) \rangle=0 \quad \forall x \in \partial \Omega\]
where $n(x)$ is the normal vector of the boundary. We emphasize that the difference between the SBP and the $L^2$-Wasserstein metric minimization \eqref{eq:BBF} is the presence of the diffusion term $\beta \Delta$ in the PDE \eqref{eq:FP}. A discussion of the viscosity limit $\beta \downarrow 0$ and the convergence of the solution to the SBP can be found in \cite{L13}.

The minimization problem \eqref{eq:functional} with PDE \eqref{eq:FP} is, as has been shown in \cite{EG99,CGP16} equivalent to minimizing the functional 
\begin{equation}
\label{eq:SBPro}
\inf_{m, \rho} \int_0^1 \int_{\Omega} \left( \frac{m(t,x)^2}{f(t,x)} + \beta^2 (\nabla \log(f(t,x)))^2 f(t,x) \right) \ dx \ dt + 2\beta \mathcal D(f_1 \vert f_0),
\end{equation}
with a constant term representing the differences of entropies $\mathcal D(f_1\vert f_0) = \int_{\Omega} f_1(x) \log(f_1(x))-f_0(x) \log(f_0(x)) \ dx$ and $f$ and $m$ are linked by the transport equation 
\[ \partial_t f(t,x) + \operatorname{div}(m(t,x))=0, \quad f(0,x)=f_0(x), f(1,x)=f_1(x).\]

The advantage of studying the functional \eqref{eq:SBPro} over \eqref{eq:BBF} is in the additional positivity and strict convexity enforced by the contribution of the Fisher information \[\mathcal I(f):=\int_{\Omega} \vert \nabla \log(f(x)) \vert^2 f(x) \ dx\]
in the objective functional. 
Numerical aspects of this minimization problem have been thoroughly discussed in \cite{LYO18}.

\subsection{Quantum optimal transport}
\label{sec:QOT}

Before introducing quantum analogues of the $L^2$-Wasserstein distance \eqref{eq:L2Wasserstein}, we first define a notion of coupling of quantum states:

For two density operators $\rho_{\mathrm{in}},\rho_{\mathrm{fi}} \in \mathscr D(\mathcal H)$ the set of all couplings $\Pi(\rho_{\mathrm{in}},\rho_{\mathrm{fi}})$ is defined as the set of density operator valued maps that smoothly (up to endpoints) connect the two states
\begin{equation}
\begin{split}
\Pi(\rho_{\mathrm{in}},\rho_{\mathrm{fi}}):=\Bigg\{& \rho \in C([0,1], \mathscr D(\mathcal H)) \cap C^{\infty}((0,1), \mathscr D(\mathcal H));\rho(0)=\rho_{\mathrm{in}}, \rho(1)=\rho_{\mathrm{fi}} \Bigg\}.
\end{split}
\end{equation}

To give the definition of the $2$-Wasserstein distance for finite-dimensional quantum systems satisfying the detailed balance equation, we employ the differential calculus introduced in \cite[Def. $4.7$]{CM20}. This framework allows us, in particular, to reformulate the evolution of finite-dimensional open quantum systems satisfying the detailed balance condition as a gradient flow of the relative entropy $S(\rho \vert \vert \sigma)$ where $\sigma$ is the invariant state, with respect to the Wasserstein metric. Before discussing this in the context of open quantum systems satisfying the detailed balance condition, we introduce the necessary differential structure:
 
\subsubsection{Differential calculus for quantum systems}
\label{subsec:DC} Let $\mathcal A$ be a finite-dimensional von Neumann algebra with faithful positive tracial linear functional $\tau$ and $\mathscr D_{+}(\mathcal A)$ the set of faithful states. 
 \begin{defi}
 A differential structure on $\mathcal A$ is defined as follows:
\begin{itemize}
\item There exists a finite index set $ J$ and for each $j \in J$ a finite-dimensional von Neumann algebra $B_j$ with a faithful positive tracial linear functional $\tau_j.$
\item For each $j \in J$ there exists a pair $(l_j,r_j)$ of unital $*$-homomorphisms from $\mathcal A$ to $B_j$ such that 
\[ \tau_j(l_j(A))= \tau_j(r_j(A))=\tau(A).\]
\item For each $j \in J$ there is $0\neq V_j \in B_j$ and $\bar j$ such that $V_j^*=V_{\bar j}$. Moreover, for $j \in J$ and $A_1,A_2 \in \mathcal A$ 
\[ \tau_j(V_j^*l_j(A_1)V_j r_j(A_2)) =  \tau_j(V_j^*r_{\bar j}(A_1)V_j l_{\bar j}(A_2)).\]
\item There is a faithful state $\sigma \in \mathscr D_{+}(\mathcal A)$ such that for each $j \in J$, $V_j$ is an eigenvector of the modular operator $M_{l_j(\sigma),r_j(\sigma)}(V_j):=l_j(\sigma)V_jr_j(\sigma)^{-1}=e^{-\omega_j} V_j$ for some $\omega_j \in \RR.$
\end{itemize}
 \end{defi}
 Then, the derivatives $\nabla_j: \mathcal A \rightarrow B_j$ are defined by $\nabla_j(A):=V_j r_j(A)-l_j(A)V_j$
 with gradient $\nabla A:=(\nabla_1 A,...,\nabla_{\vert J \vert} A)$ and divergence operator
 \[ \operatorname{div}(A) =- \sum_{j \in J} \nabla_j^* A_j\]
 where $ \nabla_j^*:=\nabla_{\bar j}$ with $\bar j$ such that $V_{\bar j}=V_j^*.$

 \subsubsection{Wasserstein distance}

 \smallsection{Logarithmic case} The quantum $L^2$-Wasserstein distance, for the above differentiable structure, has been defined in \cite[(9.1)]{CM20}, by
\[ W^2_Q(\rho_{\mathrm{in}},\rho_{\mathrm{fi}}) := \inf_{\rho \in \Pi(\rho_{\mathrm{in}},\rho_{\mathrm{fi}})}  \left\{ \int_0^1 \Vert \rho'(t) \Vert^2_{\rho(t)} dt \right\}. \]
Here, we use the norm $\Vert Z \Vert^2_{\rho}= \left\langle Z, L_{\rho}(Z) \right\rangle_{L^2(\tau)}.$
The quantum $L^2$-Wasserstein distance can then be expressed as a variational problem -in analogy to the classical Brenier-Benamou formula \eqref{eq:BBF} for the classical $L^2$-Wasserstein distance- by
\begin{equation}
\begin{split}
\label{eq:W2Q}
 W^2_Q(\rho_{\mathrm{in}},\rho_{\mathrm{fi}})&:=  \inf_{\rho \in \Pi(\rho_{\mathrm{in}},\rho_{\mathrm{fi}})}  \left\{ \int_0^1 \Vert \nabla \Phi(t) \Vert^2_{\rho(t)} \ dt \right\} 
  \end{split}
 \end{equation}
where $\Phi$ is coupled to $\rho$ by the following continuity equation
\[ \rho'(t) + \operatorname{div}\left( L_{ \rho(t)}(\nabla \Phi(t)) \right)=0.\]

The physical interpretation of the Riemannian metric $g_{\rho}$ is that for two faithful states $\rho, \sigma \in \mathscr D_{+}(\mathcal H)$, and the quantum relative entropy, defined by
 \[ S_{\sigma}(\rho) = \tau(\rho (\log(\rho)-\log(\sigma))), \]
 \cite[Prop. $2.7$]{CM20} shows that for $D$ denoting the derivative, the gradient $(\operatorname{grad}S_{\sigma})(\rho):=(-\Delta_{\rho}) DS_{\sigma}(\rho)$, where $DS_{\sigma}(\rho)= \log(\rho)-\log(\sigma)$, and we have
\[ (\mathscr L^* \rho)(\rho) = - (\operatorname{grad}S_{\sigma})(\rho).\]
This implies that the gradient flow of the entropy $S_{\sigma}$ with respect to the metric $g_{\rho}$ is the dynamics of the Liouville-von Neumann equation where $\sigma$ is the invariant state of the dynamics defined by $\mathscr L^*.$

\medskip

\smallsection{Anti-commutator case} When instead of using the the \emph{Feynman-Kubo-Mori} integral, but rather the anti-commutator 
\begin{equation}
\label{eq:anti-comm}
L^{\text{ac}}_{\rho}(T):=\frac{1}{2}\{T,\rho\}
\end{equation} one is lead to introduce a different $L^2$-Wasserstein distance \cite{CGT}
\[ \tilde{W}^2_Q(\rho_{\mathrm{in}},\rho_{\mathrm{fi}}) := \inf_{\rho \in \Pi(\rho_{\mathrm{in}},\rho_{\mathrm{fi}});}  \left\{ \int_0^1  \tr(\rho v(t)^*v(t)) dt \right\} \]
with $v^*v= \sum_{k=1}^N v_k^* v_k$, where $v$ and $\rho$ are coupled by 
\[ \rho'(t) +\operatorname{div}L^{\text{ac}}_{\rho}(v)=0:, \quad \rho(0)=\rho_{\mathrm{in}}, \rho(1)=\rho_{\mathrm{fi}}.\]

In particular, the operator $L^{\text{ac}}_{\rho}(T)$ is invertible for $\rho>0$ by standard results on the solvability of Lyapunov equations which imply that the inverse is explicitly given as
\[ (L^{\text{ac}}_{\rho})^{-1}(S) = -\int_0^{\infty} e^{-\rho s} S e^{-\rho s} \ ds.\]

 \subsection{Fermionic Fokker-Planck equation}
\label{subsec:FFP}
Due to its analogy to classical probability theory and classical gradient flows, we start by discussing the \emph{quantum fermionic Fokker-Planck equation}. Instead of just stating it within the abstract differential calculus introduced in the previous section, we will provide full details to fix ideas.

The quantum fermionic Fokker-Planck equation, is the canonical gradient flow associated with the quantum Wasserstein metric and corresponds to the classical Fokker-Planck equation\footnote{especially in statistical physics, the name \emph{Fokker-Planck equation} is usually reserved for another equation acting on phase-space variables and the equation considered here is called the (Kramers)-Smoluchowski equation.}
\[ \frac{\partial \rho(x,t)}{\partial t} = \operatorname{div}(\rho(t,x) \nabla V(x)) + \beta \Delta \rho(t,x), \quad \rho(0,x)=\rho_0(x) \text{ for } x \in \RR^d.\] 
Under suitable growth conditions on $V$ this equation has a unique invariant measure $d\mu(x) \propto e^{-\beta V(x)}dx.$
Carlen and Maas introduced in \cite{CM14} a Riemannian metric on density operators which extends the classical $L^2$-Wasserstein metric to the quantum setting and with respect to which the quantum evolution of the fermionic Fokker-Planck equation is a gradient flow. We will explain in Section how to use this metric to define a natural gradient flow for parametric models of density operators.

\subsubsection{Clifford algebra}
Let $\mathfrak C$ be the Clifford algebra on $\RR^n$ generated by $n$ self-adjoint operators $Q_j$, $j=1,..,n$ satisfying the canonical anti-commutation relations $\{Q_i,Q_j\}= 2\delta_{ij}$. The operators $Q_j$ are also called the \textit{fermionic degrees of freedom}. Moreover, $\mathfrak C$ becomes a $2^n$-dimensional Hilbert space $\mathcal H \sim L^2(\tau)$ with inner product $\langle A,B \rangle_{L^2(\tau)}:=\tau(A^*B),$ where we introduce the normalized trace $\tau(A)=2^{-n} \tr_{\CC^{2^n}}(A).$ 

The density operators $\mathscr D(\mathcal H)$ in this setting is the closed convex set of positive operators $\rho \in \mathfrak C$ of unit normalized trace.

We can explicitly construct matrices $Q_j$ solely from Pauli matrices 
\begin{equation}
\label{eq:Pauli}
 \sigma_x=\left(\begin{matrix} 0 & 1 \\ 1& 0 \end{matrix} \right), \sigma_y:= \left(\begin{matrix} 0 & -i \\ i& 0 \end{matrix} \right),\text{ and }\sigma_z=\left(\begin{matrix} 1 & 0 \\ 0 & -1 \end{matrix} \right).
 \end{equation}
 
One realization of the fermionic operators $Q_j$, is by defining them as $Q_j:=\otimes_{i=1}^n X_i$ where 
\begin{equation*}
X_i= \begin{cases} 
\sigma_z&\text{ for }i<j, \\
\sigma_x&\text{ for }i=j, \text{ and }\\
\operatorname{id}_{\CC^2} &\text{ for }i>j.
\end{cases}
\end{equation*}
The grading operator $\Gamma: \mathfrak C \rightarrow \mathfrak C$ is the linear operator defined, for $\alpha \in \{0,1\}^n$, by $\Gamma(Q^{\alpha}):=(-1)^{\vert \alpha \vert} Q^{\alpha}$ where $Q^{\alpha}:= \prod_{i=1}^n Q_i^{\alpha_i}.$ The index set $\alpha \in \{0,1\}^n$ is called the \emph{fermionic multi-index set}.
The $2^n$ matrices $Q^{\alpha}$ for $\alpha \in \{0,1\}^n$ form an orthonormal system spanning $\mathfrak C$ which satisfies $\tau(Q^{\alpha})= \delta_{0\vert \alpha \vert}.$

For two density operators $\rho_1,\rho_2 \in \mathscr D(\mathcal H)$ we define the \emph{Feynman-Kubo-Mori} operator $L_{(\rho_1,\rho_2)}: \mathcal L(\mathfrak C) \rightarrow \operatorname{TC}(\mathfrak C)$ by 
\begin{equation}
\label{eq:L}
 L_{(\rho_1,\rho_2)}(T):=  \int_0^1 \rho_1^{1-s} T \rho_2^s \ ds 
 \end{equation}
which is a contraction map into the set of trace-class operators, as H\"older's inequality shows
\[ \Vert L_{(\rho_1,\rho_2)}(T) \Vert_1 \le  \int_0^1 \Vert \rho_1^{1-s}T \rho_2^s \Vert_1  \ ds  \le \int_0^1 \Vert \rho_1^{1-s}\Vert_{(1-s)^{-1}} \Vert T \Vert_{\infty} \Vert \rho_2^s\Vert_{s^{-1}}  \ ds  \le \Vert T \Vert_{\infty}.\]
Under the stronger assumption $\rho_1,\rho_2 \in \mathscr D_+(\mathcal H),$ the operator $L_{(\rho_1,\rho_2)}$ becomes invertible and its inverse is given by \cite[Theo. $3.4$]{CM14}
\[ L_{(\rho_1,\rho_2)}^{-1}(T) = \int_0^{\infty} (\rho_1+t)^{-1} T (\rho_2+t)^{-1} \ dt. \]
In particular, we will just write $L_{\rho}:= L_{(\Gamma(\rho),\rho)}$ in the sequel.

The fermionic Dirichlet form on $\mathfrak C$ is defined by 
\[ \mathcal F(A,A) =\tau((\nabla A)^* \nabla A)= \sum_{j=1}^n \tau ( (\nabla_j A)^* \nabla_j A) \]
with derivatives 
\begin{equation}
\label{eq:derivatives}
\nabla_j(A) = \frac{1}{2} \left(Q_j A - \Gamma(A) Q_j\right) \in \mathfrak C,\text{ for }j \in \{1,..,n \}\text{ and }A \in \mathfrak C.
\end{equation}
The gradient $\nabla: \mathfrak C \rightarrow \mathfrak C^n$ is then defined as $\nabla(A) := \left( \nabla_1(A),...,\nabla_n(A)\right) \in \mathfrak C^n$ with nullspace $\operatorname{ker}(\nabla) = \operatorname{span}(\operatorname{id}).$ The $L^2(\tau)$-adjoint of derivatives $\nabla_j$ is just given by 
\[ \nabla_j^*(A) = \frac{1}{2} \left(Q_j A + \Gamma(A) Q_j\right).\]

The divergence operator is defined, for $A=(A_j)_j \in \mathfrak C^n$ by $\operatorname{div}(A) =-\sum_{j=1}^n\nabla_j^*(A_j)$. We define the fermionic number operator $\mathcal N$ as the self-adjoint operator associated to the Dirichlet form $\mathcal F(B,A) =: \langle B, \mathcal N A \rangle_{L^2(\tau)}$
where $\mathcal NA = -\operatorname{div}(\nabla A)$ for all $A \in \mathfrak C$ and $\operatorname{ker}(\mathcal N)=\operatorname{id}.$ The dynamical semigroup generated by $-\mathcal N$ is the \emph{quantum fermionic Fokker-Planck semigroup} defined by $P_t=e^{-t\mathcal N}$ which relaxes exponentially fast to its unique invariant state, the completely mixed state. In particular, $\mathcal N$ is the generator of an ergodic Quantum Markov semigroup satisfying the detailed balance condition with respect to the completely mixed state.

This model can be casted in the differential calculus introduced in Section \ref{subsec:DC} by setting $\mathcal A:=\mathcal B_j:=\mathfrak C^n$, $V_j:=Q_j$, $\omega_j:=0$, $l_j:=\Gamma$ and $r_j=\operatorname{id}$ with derivatives as defined in \eqref{eq:derivatives}
and a generator $\mathscr LA = 2 \sum_{j=1}^n (Q_jA Q_j-A)=-4\mathcal N.$

 \subsection{Quantum Markov semigroups with detailed balance condition}
 \label{sec:QMSDB}
 In the rest of this section, we illustrate the ideas using the differential calculus in Subsection \ref{subsec:DC} in the case of Quantum Markov semigroups $(P_t)$ with Lindblad generator $\mathscr L$, in the Heisenberg picture, acting on a finite-dimensional $C^*$-algebra $\mathcal A$ satisfying the detailed balance condition (DBC). This means, that for all times $t>0$ the operator $P_t$ is self-adjoint with respect to the inner product $\langle X, Y\rangle_{1,\sigma}:=\tau(X^* \sigma Y)$ for some state $\sigma$. In particular, the DBC implies that $\sigma$ is the unique state such that $P_t^*(\sigma)=\sigma$ for all times $t>0.$ 
Other possible applications of the differential calculus in Subsection \ref{subsec:DC} and thus also of the parameter estimation techniques studied in this paper are discussed in \cite[Sec. $5$]{CM20} and include popular quantum channels such as the depolarizing channel.

The generators $\mathscr L$ of the quantum Markov semigroups in Heisenberg representation are characterized by \cite[Theo $2.4$]{CM20}
\begin{equation}
\label{eq:generators}
 \mathscr L = \sum_{j \in J} e^{-\omega_j/2} \mathscr L_j \text{ and } \mathscr L_j (A) = V_j^* [A,V_j]+ [V_j^*,A]V_j
 \end{equation}
 with $J$ a finite set and a family of operators $(V_j)_{j \in J}$ closed under taking adjoints, as well as real numbers $\omega_j$ such that the modulation operator $M_{\sigma}(A):=M_{\sigma,\sigma}(A):=\sigma A \sigma^{-1}$ satisfies 
 \[ M_{\sigma}(V_j) = e^{-\omega_j } V_j \text{ and }\omega_{\bar j} =-\omega_j. \] 
 We then define $\mathcal A= B_j= \mathcal L(\mathcal H)$ where $\mathcal H$ is a finite-dimensional Hilbert space, write $\mathcal B:=\prod_j B_j$, and set $l_j=r_j=\operatorname{id}_{\mathcal A}.$ The partial derivatives are then just given by $\nabla_j A = [V_j,A]$ and $\nabla_j^*:=\nabla_{\bar j}$ where $\bar j$ is such that $V_j^*=V_{\bar j}.$ The gradient vector is thus just $\nabla = (\nabla_1,...,\nabla_{\vert J \vert})$.  It follows from \cite[Prop. $2.5$]{CM20} that the Lindblad generator induces a Dirichlet form with respect to the Kubo-Martin-Schwinger inner product $\langle A,B \rangle_{\mathrm{KMS}}:=\tau(X^* Y \sigma),$ i.e.
 \[ \langle \nabla A, \nabla B \rangle = - \langle A, \mathscr L B \rangle_{L^2_{\mathrm{KMS}}(\sigma)} \text{ for all } A,B \in \mathcal A. \]
 
 We then define the operator
 \[ \widehat{\rho_j} = \int_0^1 (e^{\omega_j/2}l(\rho))^{1-s} \otimes  (e^{-\omega_j/2}r(\rho))^{s} \ ds \in \mathcal A \otimes \mathcal A \]
 with inverse operator
  \[ \widecheck{\rho_j} = \int_0^{\infty} (t+e^{\omega_j/2}l(\rho))^{-1} \otimes   (t+e^{-\omega_j/2}r(\rho))^{-1} \ dt. \]
In terms of a contraction operator $\#$ that is uniquely defined as the linear extension of the map $(A \otimes B)\#C:=ACB$ for $A,B,C \in \mathcal A$ and \emph{Feynman-Kubo-Mori} operator
\begin{equation}
\label{eq:L2}
L_{\rho}(C) := \widehat{\rho_j} \# C
\end{equation} we may then introduce a positive-definite operator $-\Delta_{\rho}$ on $L^2(\mathcal A, \tau)$
\begin{equation}
\label{eq:D2}
 -\Delta_{\rho}(A) :=\sum_{j \in J} \nabla_j^*(L_{\rho}(\nabla_j A)).
 \end{equation}
This way, the $L^2$-quantum Wasserstein metric becomes
 \begin{equation}
\begin{split}
 W^2_Q(\rho_{\mathrm{in}},\rho_{\mathrm{fi}})&:=  \inf_{\rho \in \Pi(\rho_{\mathrm{in}},\rho_{\mathrm{fi}})}  \left\{ \int_0^1 \langle \Phi(t),-\Delta_{\rho} (\Phi(t)) \rangle_{\tau} \ dt \right\} 
  \end{split}
 \end{equation}
where $\Phi$ is coupled to $\rho$ by the following continuity equation
\[ \rho'(t) = \Delta_{\rho(t)} \Phi(t).\]

 \section{Quantum natural gradient and open quantum systems}
 \label{sec:DBC}
In the following we shall impose the following condition on generators of finite-dimensional open quantum systems we consider:
 \begin{ass}
We assume that $\mathscr L$ is ergodic, i.e. $\operatorname{ker}(\mathscr L)=\operatorname{span}\{ \operatorname{id}\}$ satisfying the detailed balance condition with invariant state $\sigma.$  \end{ass} 

 \subsection{Gradient flow for finite-dimensional OQSs with DBC}
 \label{eq:GFFF}
By the ergodicity assumption, we are able to pull back the metric from the state space to the parameter space. In particular, the above assumptions are satisfied for the fermionic Fokker-Planck equation with the completely mixed state as the unique invariant state.

The statistical parameter space is as in the classical setting defined as a $d$-dimensional Riemannian manifold $\Theta$ with connection $D_{\theta}$ and metric tensor $\langle \xi, \eta \rangle_{\theta} = \xi^T G_{\theta} \eta$. We then take a continuous parametrization $\Theta \ni \theta \mapsto \rho(\theta) \in \mathscr D_{+}(\mathcal A)$ of density operators. 
 
We then define a norm 
\[\Vert Z \Vert^2_{\rho}= \left\langle Z, L_{\rho}(Z) \right\rangle_{L^2(\tau)}\]
where $L_{\rho}$ has been defined in \eqref{eq:L} for the fermionic Fokker-Planck equation and in \eqref{eq:L2} for general open quantum systems satisfying the DBC. In addition, we allow for $L_{\rho}$ the anti-commutation operator defined in \eqref{eq:anti-comm}.

The associated metric tensor on $\mathscr D_+(\mathcal H)$ is given by
\[ g_{\rho}: (T_{\rho}\mathscr D_+(\mathcal H))^2 \rightarrow \RR, \quad g_{\rho}(X,X):= \left\langle \nabla \Phi_X, L_{\rho}(\nabla \Phi_X) \right\rangle_{L^2(\tau)} \]
where $T_{\rho}\mathscr D_+(\mathcal H)$ is the tangent space at $\rho$ and $\nabla \Phi_X$ is the unique gradient field  cf. \cite[Theo $3.17$]{CM14} and \cite[Lem. $7.5$]{CM20} satisfying 
\[ X = - \operatorname{div} \left( L_{\rho}(\nabla \Phi_X) \right).\]
In case of $L_{\rho}$ being the anti-commutator, the gradient field $\nabla \Phi_X$ can be found by solving the Lyapunov equation \cite[(21)]{CGT}
\[\nabla ( \operatorname{div}\operatorname{grad}\vert_{\operatorname{span}(\operatorname{id})^{\perp}} )^{-1} X = L_{\rho}(\nabla \Phi_X) \in \mathcal L(\mathcal H^n).\]

In particular, there exists a unique gradient $\nabla \Phi_{\xi}$ such that 
\[ \langle D_{\theta}\rho(\theta),\xi \rangle_{\theta} =-  \operatorname{div}(L_{\rho(\theta)} (\nabla \Phi_{\xi})).\]

Hence, we conclude that for $\xi,\eta \in T_{\theta}\Theta$ there are \emph{score functions} $\Phi_{\xi}$ and $\Phi_{\eta}$\footnote{Score functions are only defined up to elements in $\operatorname{ker}(\nabla)$} such that we can define the pullback metric on the parameter space
\begin{equation}
\begin{split}
\label{eq:MT2}
 g_{\theta}(\xi,\eta)&:=g_{\rho(\theta)}(\langle D_{\theta}\rho(\theta),\xi \rangle_{\theta},\langle D_{\theta}\rho(\theta),\eta \rangle_{\theta})\\
 &=  \left\langle \nabla \Phi_{\xi}, L_{\rho(\theta)} (\nabla \Phi_{\eta}) \right\rangle_{L^2(\tau)}\\
 &=-\left\langle  \Phi_{\xi}, \operatorname{div}(L_{\rho(\theta)} (\nabla \Phi_{\eta})) \right\rangle_{L^2(\tau)} \\
  &=\left\langle \Phi_{\xi}, \langle D_{\theta} \rho(\theta),\eta \rangle_{\theta} \right\rangle_{L^2(\tau)}.
 \end{split}
\end{equation}

We then define the operator $-\Delta_{\theta}f:=-\operatorname{div}(L_{\rho(\theta)}(\nabla f)).$ This operator is self-adjoint with respect to $\langle \bullet, \bullet \rangle_{L^2(\tau)}$ and positive-definite with only $\operatorname{span}\{\operatorname{id}\}$ in its nullspace by ergodicity. Using that $\langle D_{\theta} \rho(\theta),{\eta} \rangle_{\theta} \in \operatorname{ker}(\Delta_{\theta})^{\perp}$, this implies that
\begin{equation}
\begin{split}
 g_{\theta}({\xi},{\eta})&:= \tau \left(  \langle D_{\theta} \rho(\theta),{\xi} \rangle_{\theta}(- \Delta_{\theta})\vert_{\operatorname{span}\{\operatorname{id}\}^{\perp}}^{-1} \langle D_{\theta} \rho(\theta),{\eta} \rangle_{\theta} \right).
 \end{split}
 \end{equation}

We can rewrite this line as a bilinear form by using the matrix $G_{\theta}$, introduced above, associated with the metric on the parameter space. We can thus define the positive definite \emph{Wasserstein information matrix}
\begin{equation}
\label{eq:GW}
G_W(\theta)=\tau\left(\widehat{e_i}^T G_{\theta}^T D_{\theta}\rho(\theta)(-\Delta_{\theta})\vert_{\operatorname{span}\{\operatorname{id}\}}^{-1}  D_{\theta} \rho(\theta)G_{\theta} \widehat{e_j} \right)_{i,j} \in \RR^{d \times d}.
\end{equation}
 Thus, it follows that the metric tensor of the pullback metric on the statistical manifold is of the simple form
 \begin{equation}
 \label{eq:MT3}
 g_{\theta}(\xi,\eta)= \langle \xi, G_W(\theta) \eta \rangle
 \end{equation}
and as in Section \ref{sec:NGF} the natural Wasserstein gradient becomes for an objective function $R(\theta)$ defined by 
\begin{equation}
\label{eq:WGN}
 \dot{\theta}(t) = -\nabla_{g} R(\theta(t))
 \end{equation}
with $\nabla_g$ uniquely defined by
\[ g_{\theta}(\nabla_g R(\theta), \xi ) = \langle D_{\theta} R(\theta), \xi \rangle_{\theta} \quad \forall \xi \in T_{\theta}\Theta\]
such that $\nabla_g R(\theta) = G_W(\theta)^{-1} G_{\theta} D_{\theta} R(\theta).$
This is illustrated in Section \ref{sec:WNG} for $R$ being the von Neumann entropy.

The gradient descent method in parameter space naturally corresponds to a gradient descent method on the parametrized manifold of states:

\begin{prop}
Consider an immersion $\Theta \ni \theta \mapsto \rho(\theta) \in \mathscr D(\mathcal H)$ and an objective function $\mathcal R$ on the set of states. We can then define an objective function $R(\theta)=\mathcal R (\rho(\theta))$ and the gradient evolution
\[ \dot{\theta}(t) = - \nabla_g R(\theta),\]
induces the gradient evolution 
\[ \rho'(t) = - \operatorname{grad} \mathcal R(\rho(t))\]
on the parametrized manifold of states where $\rho(t)=\rho(\theta(t))$ and $\operatorname{grad}(\mathcal R(\rho(t_0)))=\langle D_{\theta}\rho(\theta),\nabla_g R(\theta_{t_0}) \rangle_{\theta}.$
\end{prop}
\begin{proof}
We always have that $\frac{d}{dt}\rho(\theta(t))=(\rho_{\theta})_{*} \dot{\theta}(t)= - (\rho_{\theta})_{*}\nabla_g R(\theta(t)).$ Thus it suffices to show that $ (\rho_{\theta})_{*}\nabla_g R(\theta(t))=  \operatorname{grad} \mathcal R(\rho(t)).$

Fix a curve $(\vartheta_{\tau})_{\tau}$ passing through $\theta_{t_0}$ at $\tau=0$, then it follows that 
\begin{equation}
\begin{split}
\frac{d}{d\tau} \big\vert_{\tau=0} R(\vartheta_{\tau})&= g_{\theta_{t_0}}(\nabla_g R(\theta_{t_0}),\dot{\vartheta}_{0})\\
&=g_{\rho({t_0})}(\langle D_{\theta}\rho(\theta_{t_0}),\nabla_g R(\theta_{t_0}) \rangle_{\theta}, \langle D_{\theta}\rho(\theta_{t_0}),\dot{\theta}_{t_0}) \rangle_{\theta}).
 \end{split}
 \end{equation}
On the other hand, we also see that
\begin{equation}
\begin{split} \frac{d}{d\tau} \big\vert_{\tau=0} R(\vartheta_{\tau})&=\frac{d}{d\tau} \big\vert_{\tau=0} \mathcal R(\rho(\tau))= g_{\rho({t_0})}(\operatorname{grad}(\mathcal R (\rho(t_0))),\dot{\rho}(\theta(0))) \\
&=g_{\rho({t_0})}(\operatorname{grad}(\mathcal R (\rho(t_0))),\langle D_{\theta}\rho({t_0}),\dot{\theta}(0) \rangle_{\theta}).
 \end{split}
 \end{equation}
This shows that $\langle D_{\theta}\rho(\theta),\nabla_g R(\theta_{t_0}) \rangle_{\theta}=\operatorname{grad}(\mathcal R(\rho({t_0}))).$
\end{proof}

Using \eqref{eq:MT2} and \eqref{eq:MT3}, we thus find that the geodesics on the parameter manifold $(\Theta, g_{\theta})$ minimize again the square geodesic distance
 \begin{equation}
W^2_Q(\rho(\bullet,\theta^0),\rho(\bullet,\theta^1)) = \inf_{\substack{\theta \in C^1(0,1)\cap C[0,1]\\ \theta(0)=\theta^0 , \theta(1)=\theta^1}}\left\{ \int_0^1 \langle \dot{\theta}(t), G_W(\theta(t)) \dot{\theta}(t)\rangle \ dt \right\}.
 \end{equation}
The geodesics to the above Wasserstein distance are given as solutions to the following Hamiltonian system
 \begin{equation}
 \begin{split}
 \label{eq:low-dim}
&\dot{\theta}-G_W(\theta)^{-1} P = 0 \text{ and }\\
&\dot{P}+ \frac{1}{2} P^T \partial_{\theta} G_W(\theta)^{-1} P=0. 
 \end{split}
 \end{equation}
Indeed, for the Lagrangian $\mathcal L(\theta(t), \dot{\theta}(t)) =\langle \dot{\theta}(t),G_W(\theta(t))\dot{\theta}(t) \rangle,$ the associated momentum variable is $P(t)=G_W(\theta(t))\dot{\theta}(t)$ with Hamilton function $H(P(t),\theta(t)) = \frac{1}{2} \langle P(t), G_W(\theta(t)) P(t)\rangle$. The system \eqref{eq:low-dim} are then precisely Hamilton's equations.

 On the other hand, the geodesic equations in $\mathscr D_+(\mathcal H)$ with respect to the quantum Wasserstein metric for the fermionic Fokker-Planck equation are given by \cite[Theo. $5.3$]{CM14}
 \begin{equation}
 \begin{split}
 \label{eq:high-dim}
 &\rho'(t) + \operatorname{div}(L_{\rho(t)} \nabla \Phi(t) ) = 0 \\
 &\Phi'(t) + \frac{1}{2} \rho(t) \flat (\nabla \Phi(t), \nabla \Phi(t)) = 0
 \end{split}
 \end{equation}
 where we define for $\rho \in \mathscr D_+(\mathcal H)$ and $X,Y \in \mathfrak C^n$ the map 
 \begin{equation}
 \begin{split}
  \rho \flat (X,Y)  =  \int_0^1 \int_0^1 \int_0^{\alpha} \frac{2\rho^{\alpha-\beta}}{(1-s)+s\rho} X^* \Gamma(\rho)^{1-\alpha} Y \frac{\rho^{\beta}}{(1-s)+s\rho} \ d\beta \ d\alpha \ ds.
 \end{split}
 \end{equation}
 The advantage of \eqref{eq:low-dim} over \eqref{eq:high-dim} lies in the low-dimensionality of the parameter space which turns \eqref{eq:low-dim} into an equation in a much lower dimensional space than the system in \eqref{eq:high-dim}, in general.  
 
 \subsection{Schr\"odinger bridge problem for finite-dimensional OQSs with DBC}
 \label{sec:QSBP}
We may now introduce a generalization of the quantum Brenier-Benamou formula in \eqref{eq:W2Q}, to study a quantum version of the Schr\"odinger bridge problem, by adding a Fisher information regularizer to the dynamics. For this derivation, we shall restrict us to the scenario that the operator $L_{\rho}$ is the \emph{Feynman-Kubo-Mori} operator as in this case, one obtains direct links to quantum entropies and quantum dynamics. 

The computational advantage of the Fisher information regularization are two-fold. Firstly, it induces additional convexity to the minimization problem. Secondly, it additionally forces the density operator to remain strictly positive. 
 \begin{defi}
The quantum Schr\"odinger bridge problem (SBP) is the minimization problem for two quantum states $\rho_{\mathrm{in}},\rho_{\mathrm{fi}} \in \mathscr D_+(\mathcal H)$ 
\begin{equation}
\label{eq:SBP}
\mathcal S(\rho_{\mathrm{in}},\rho_{\mathrm{fi}}):=\inf_{\rho \in \Pi(\rho_{\mathrm{in}},\rho_{\mathrm{fi}}) } \inf_m \int_0^1  \Vert m \Vert^2_{\rho(t)^{-1}} \ dt.
\end{equation}
 where we use the inner product
 \[ \langle X, Y \rangle_{\rho(t)^{-1}}:= \langle X, L_{\rho}^{-1}(Y) \rangle_{L^2(\tau)}. \]
  Here $m$ is connected to $\rho(t)$ by an inhomogeneous heat equation 
 \[ \rho'(t) + \operatorname{div}(m(t)) = \beta T \rho(t) \]
 for some fixed parameter $\beta\ge 0$  where $T=\mathscr L^*$ for OQS satisfying the DBC and $T=-\mathcal N$ in the case of the fermionic Fokker-Planck equation.
 \end{defi}
As for \eqref{eq:SBP}, in the case $\beta=0$, the SBP reduces to the minimization of the quantum $L^2$-Wasserstein metric in \eqref{eq:W2Q}.
 We now introduce the \emph{Fisher information matrix}  $I(\rho):=  \Vert \nabla (\log(\rho)-\log(\sigma))\Vert^2_{\rho}.$ We can then express the optimal transport distance problem \eqref{eq:SBP} as an equivalent dynamical problem with Fisher information regularization:
 \begin{theo}
 The Schr\"odinger bridge problem \eqref{eq:SBP} is equivalent to the following optimization problem 
 \begin{equation*}
 \mathcal S (\rho_{\mathrm{in}},\rho_{\mathrm{fi}}):=\inf_{\rho \in \Pi(\rho_{\mathrm{in}},\rho_{\mathrm{fi}}) } \inf_M \int_0^1  \Vert M(t) \Vert^2_{\rho(t)^{-1}} + \beta^2 I(\rho(t)) \ dt + 2\beta (S_{\sigma}(\rho_{\mathrm{fi}})-S_{\sigma}(\rho_{\mathrm{in}}))
 \end{equation*}
 where $M$ satisfies the transport equation
 \begin{equation}
\label{eq:transport2}
\rho'(t) + \operatorname{div} M(t) = 0.
\end{equation}
 \end{theo}
 \begin{proof}
We start by defining
 \[ M(t):=m(t)-\beta  L_{\rho(t)}(\nabla( \log(\rho(t))-\log(\sigma))) ,\]
which turns the inhomogeneous heat equation into a simple transport equation
\begin{equation}
\rho'(t) + \operatorname{div} M(t) = 0
\end{equation}
as 
\[\mathscr L^* =\operatorname{div}(L_{\rho(t)}(\nabla( \log(\rho(t))-\log(\sigma)))),\] cf. the proof of \cite[Prop. $2.7$]{CM20}. In case of the fermionic Fokker-Planck equation we also record that the quantum analog of the classical identity $\nabla f(x) = f(x) \nabla \log f(x)$ in the quantum setting becomes the identity \cite[Lemma $3.1$]{CM14}
 \[ \nabla_i \rho = L_{\rho}(\nabla_i \log(\rho)).\]
Thus, we have that
 \begin{equation}
 \begin{split}
 \label{eq:reduction2}
\Vert m(t) \Vert^2_{\rho(t)^{-1}}&= \left\lVert M(t)+\beta L_{\rho(t)}(\nabla( \log(\rho(t))-\log(\sigma))) \right\rVert^2_{\rho(t)^{-1}} \\
&=  \left\lVert M(t) \right\rVert^2_{\rho(t)^{-1}} + 2 \beta \langle M(t), \nabla (\log(\rho(t))-\log(\sigma)) \rangle_{L^2(\tau)} \\
&\qquad + \beta^2 \Vert \nabla (\log(\rho(t))-\log(\sigma))\Vert^2_{\rho(t)}.
\end{split}
 \end{equation}
 
The middle term in \eqref{eq:reduction2} is constant, and satisfies in terms of the relative von Neumann entropy $S_{\sigma}(\rho) = \tau (\rho (\log(\rho)-\log(\sigma)))$ 
  \begin{equation*}
 \begin{split}
\int_0^1 \langle M(t), \nabla (\log(\rho(t))-\log(\sigma)) \rangle_{L^2(\tau)} \ dt 
&= -\int_0^1 \tau(\operatorname{div}(M(t)) (\log(\rho(t))-\log(\sigma)) ) \ dt\\
&=\tau\left( \int_0^1 \rho'(t)  (\log(\rho(t))-\log(\sigma)) \ dt \right)\\
&=S_{\sigma}(\rho_{\mathrm{fi}})-S_{\sigma}(\rho_{\mathrm{in}})= \operatorname{const}
\end{split}
 \end{equation*}
where we integrated by parts to obtain the last line. 
\end{proof}

 \subsection{Continuous-variable systems}
 \label{sec:CVS}
As in the theory of classical probability theory, there exists a close analogue of quantum \emph{Gaussian states} $\mathcal G(\mathcal H_m)$ on $\mathcal H_m:=L^2(\RR^m)$ defined as follows (cf. \cite{BDLR19} and references therein for more details):

Gaussian states are states $\rho \in \mathscr D(\mathcal H_m)$ such that their characteristic function $\chi_{\rho}: \CC^m \rightarrow \CC$
\begin{equation}
\label{eq:charfunction}
\chi_{\rho}(z):= \tr(\rho D(z)) 
\end{equation}
is the characteristic function of a Gaussian random variable over $\CC^m$, i.e. $\chi (\xi)=\operatorname{exp}\left(-\frac{1}{4} \langle \xi,\gamma \xi \rangle+i \langle d, \xi \rangle \right)$ where $\gamma>0$ is a positive definite matrix satisfying $\gamma+i\nu \ge 0$, for $\nu:=\begin{pmatrix} 0 & 1 \\ -1 & 0 \end{pmatrix}^{\oplus_{i=1}^m}$, and $d \in \RR^{2m}.$ Here, $D(z)$ is the displacement operator
\[ D(z) = \operatorname{exp}\left(\sum_{j=1}^m (z_j a_j^* - \bar z_j a_j)\right).\]
Conversely, the density operator $\rho \in \mathscr D(\mathcal H_m)$ can be recovered from its characteristic function by 
\[ \rho = \int_{\CC^m} \chi_{\rho}(z)D(-z) \frac{dz}{\pi^m}.\]
We can associate a canonical random variable to any Gaussian state in terms of their Wigner function 
\begin{equation}
\label{eq:Wigner}
 P_{\rho}(z):=\int_{\CC^m} \chi_{\rho}(w) e^{z^T w^*-z^{\dagger} w } \frac{dw}{\pi^{2m}}\ge 0
 \end{equation}
which is of unit $L^1$ norm and a Gaussian distribution on $\RR^{2m}$ as well.

A particularly simple and relevant example of a Gaussian state are thermal states with mean photon number $N \in [0,\infty)$ 
\[ \rho_N:= \frac{1}{N+1}\sum_{n=0}^{\infty} \left(\frac{N}{N+1} \right)^n \vert n \rangle \langle n \vert \]
as their characteristic functions and Wigner distributions 
\begin{equation}
\begin{split}
\chi_{\rho_N}(z):=e^{-(2N+1) \vert z \vert^2/2} \text{ and }P_{\rho_N}(z):=\frac{2}{\pi(2N+1)} e^{-\frac{2}{2N+1} \vert z \vert^2}.
\end{split}
\end{equation}
are centered and uncorrelated.

Thermal states have the special property that they are the maximum entropy states for a fixed average energy
\[ \rho_N = \operatorname{argmax}_{\rho ; \tr(\rho a^*a)\le N} -\tr(\rho \log(\rho)).\]

We finally mention that although Wigner distributions functions are positive as operators on $L^2(\RR^{2m})$, they are not pointwise positive in general and therefore also not always genuine probability distributions (cf.the Wigner distribution function associated to $\vert 1 \rangle \langle 1 \vert$).

In addition, the Wigner distribution function of a state $\rho$ satisfies the energy identity
\[ \int_{\RR^{2n} } \vert z\vert^2 P_{\rho}(z) \ dz = \int_{\RR^{2n}} \vert z \vert^2 \rho(z) \ dz = \tr(\rho x^2) + \tr(\rho p^2) = \tr((2a^*a+1)\rho) \]
where $x$ and $p$ are the position and momentum operator. 

Thus, the classical $L^2$-Wasserstein distance, corresponds in this formalism to an energy penalization and we define the optimal transport functional with phase-space variable square penalization
\[ \inf_{m. \rho} \int_0^1 \int_{\RR^{2n}} \frac{\vert m(t,z) \vert^2}{\rho(t,z)} \ dz \ dt \]
where $\rho$ satisfies the Fokker-Planck equation
\[ \partial_t \rho(t,z) + \operatorname{div}(m(t,z)) = \beta \Delta \rho(t,z) , \quad \rho(0,z)=\rho_0(z), \ \rho(1,z)=\rho_1(z)\]
with parameter $\beta  \ge 0,$ where $\beta=0$ corresponds to the optimal transport in $L^2$-Wasserstein distance \cite{CL18} and $\beta>0$ to the Schr\"odinger bridge problem \cite{LYO18}. 

\begin{prop}[Separability]

Let $\rho^{(i)}_{\theta}$ be a family of Gaussian states on Hilbert spaces $L^2(\RR^{2n(i)})$, and $\rho_{\theta}:=\bigotimes_{i=1}^N \rho^{(i)}_{\theta}$, then the Wasserstein information matrix satisfies 
\[ G_W(\rho)= \sum_{i=1}^N G_W(\rho^{(i)}).\]
\end{prop}

\begin{proof}
It follows directly from \eqref{eq:charfunction} that the characteristic function of a tensor product is the product of the individual characteristic functions. Using the Fourier transform and \eqref{eq:Wigner}, this immediately translates into the Wigner functions being a product of Wigner functions \eqref{eq:Wigner}. The result then follows from \cite[Prop. $5$]{LZ20}.
\end{proof}

\section{Examples}
\label{sec:Ex}
In this section, we demonstrate the quantum transport information matrix and its related gradient and Hamiltonian flows in some well-known probability models. 

\subsection{Examples for the quantum Wigner distribution}
\subsubsection{Gaussian mixture model} 

For Gaussian states $\rho_i$ we consider the Gaussian Wigner probability distributions $P_{\rho_i}$ associated to them. 

Let $X_i \sim \mathcal N(\mu_i, \Sigma_i)$ be normal random variables, then it follows that 
\[ \mathbb E(X_i) = m_i, \text{Var}(X_i) = \mathbb E(X_iX_i^* ) -\mathbb E(X_i)\mathbb E(X_i)^* = \Sigma.\]

Let $X= \sum_{i=1}^N \lambda_i X_i$ be a Gaussian mixture with $\lambda_i \ge 0$ summing up to one, then clearly
$\mu_X:=\mathbb E(X)=\sum_{i=1}^N \lambda_i \mu_i$ and also for the second moment $m_{X_i}:=\mathbb E(X_iX_i^*)$ we find 
\[ m_{X}:=E_X(xx^*)= \sum_{i=1}^N \lambda_i \mu_{X_i}.\]
Thus, the covariance matrix is given by 
\[ \text{Var}(X) = \sum_{i=1}^N \lambda_i \Sigma_i + \sum_{i=1}^N \lambda_i \mu_i \mu_i^* -\mathbb E(X) \mathbb E(X)^*\]
where $\sum_{i=1}^N \lambda_i \mu_i \mu_i^* -\mathbb E(X) \mathbb E(X)^* \ge 0$ by Jensen's inequality.
Thus, since the variance of a mixture is increasing, the condition $\Sigma_i + i \nu \ge 0$ is satisfied for the extremal states and clearly the state associated with the mixture $X$ is 
\[\rho = \sum_{i=1}^N \lambda_i \rho_i. \]

To parametrize multivariate Gaussian distributions $\mathcal N(\mu, \Sigma)$ that are Wigner functions of Gaussian states, it is natural to consider the parameter space $\theta=(\mu, \Sigma) \in \Theta:= \RR^{2m} \times \{\gamma \in \RR^{2m \times 2m};\gamma >0 \text{ and } \gamma+i \nu > 0\}.$ 
The Wasserstein metric tensor for the multivariate Gaussian model is 
\[ g_{\theta}(\xi,\eta) = \langle \mu_{\xi}, \mu_{\eta}  \rangle+ \tr(S_{\xi} \Sigma S_{\eta})\]
for $\xi= (\mu_{\xi}, \Sigma_{\xi})$ and $\eta =(\mu_{\eta}, \Sigma_{\eta})$ and $S_{\xi}$ and $S_{\eta}$ solving the Lyapunov equations 
\[ \Sigma_{\xi}=\{S_{\xi}, \Sigma\} \text{ and } \Sigma_{\eta} =\{S_{\eta}, \Sigma\}.\]

In fact, for $Q =Q^*$, we can define the map $L_{\Sigma}(Q):=\int_0^{\infty} e^{-\Sigma t} Q e^{-\Sigma t} \ dt,$ solving Lyapunov equation $Q=\{L_{\Sigma}(Q), \Sigma\}$, then $L_{\Sigma}(\Sigma_{\eta})= S_{\eta}$ and $L_{\Sigma}(\Sigma_{\xi})= S_{\xi}.$

This way, setting $G_W:=\indic_{\RR^{2n}} \oplus (L_{\Sigma_{\eta}} \Sigma L_{\Sigma_{\eta}})$ we find that 
\[ g_{\theta}(\xi, \eta)= \langle (\mu_{\xi}, \Sigma_{\xi}), G_W  (\mu_{\eta}, \Sigma_{\eta}) \rangle.\]

\begin{ex}[Gaussian states; Numerical solution]
\label{ex:Wigner1}
We consider two Gaussian states with associated Wigner distributions and parameters $\theta^0:=(\Sigma_0,\mu_0)$ and $\theta^1:=(\Sigma_1,\mu_1)$
\begin{equation}
\begin{split}
&\Sigma_{0}:=\begin{pmatrix} 26 & 1 \\ 1 & 1 \end{pmatrix} \text{ and } \Sigma_{1}:=\begin{pmatrix} 1 & 1 \\ 1 & 2 \end{pmatrix}.
 \end{split}
 \end{equation}
which are easily shown to satisfy $\Sigma + i\nu \ge 0$ and expectation values
\begin{equation}
\begin{split}
\mu_{0}:=(-1,-1)^t, \text{ and }  \mu_{1}:=(2,7)^t.
\end{split}
\end{equation}
For Wigner functions
\[ W(\Sigma,\mu)(x,\xi)= \frac{ e^{-\frac{1}{2}\langle (x,\xi)^t-\mu, \Sigma^{-1}((x,\xi)^t-\mu) \rangle} }{2\pi \sqrt{\vert \Sigma \vert}} \]
we then want to analyze the optimal transport plan between $W(\theta^0)$ and $W(\theta^1)$.

\begin{figure}
\centering
\begin{minipage}[b]{0.45\textwidth}
\includegraphics[width=\linewidth]{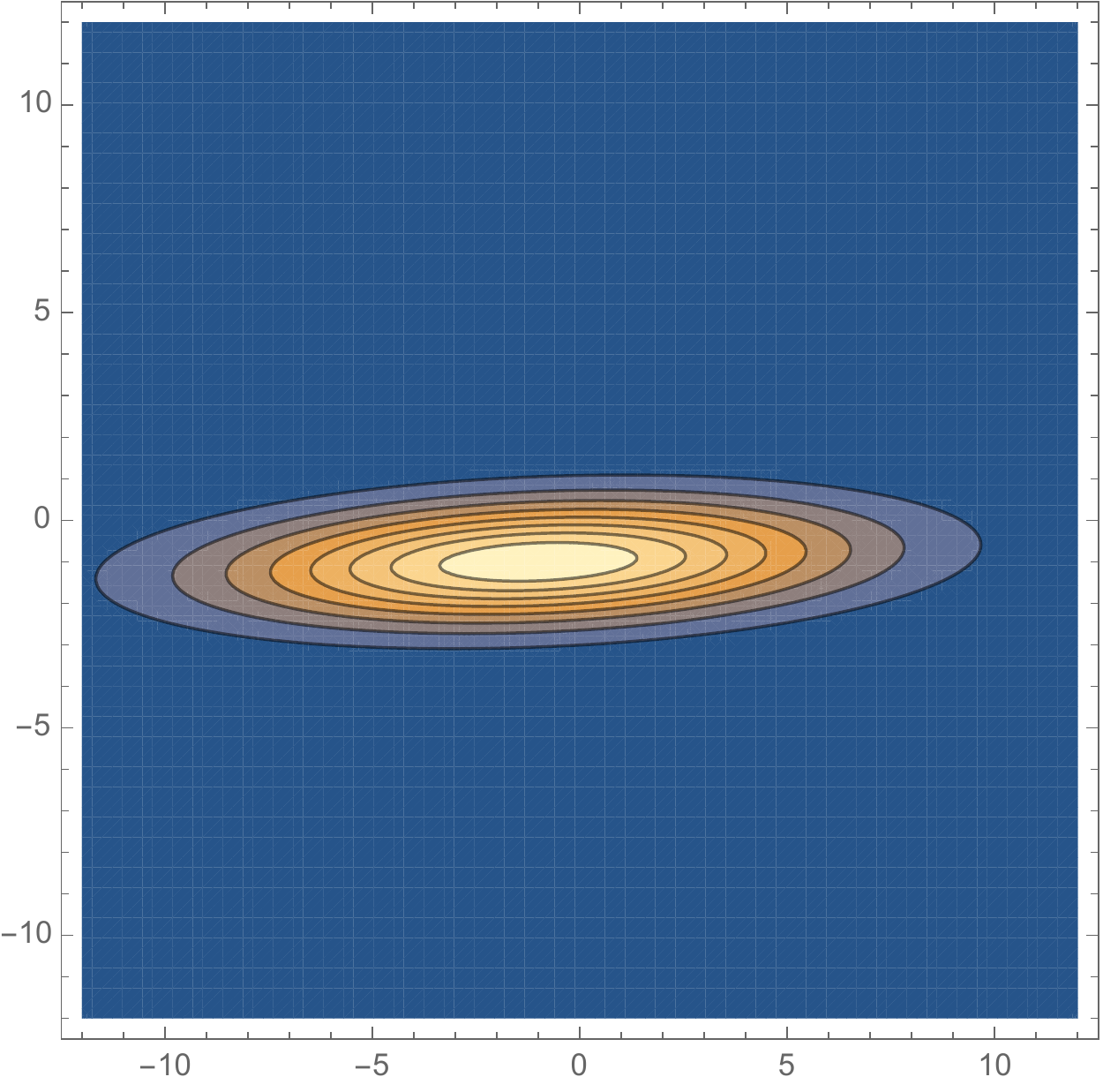}
\caption{The Wigner function $W(\Sigma_0, \mu_0)$.}
\end{minipage}
\begin{minipage}[b]{0.45\textwidth}
\includegraphics[width=\linewidth]{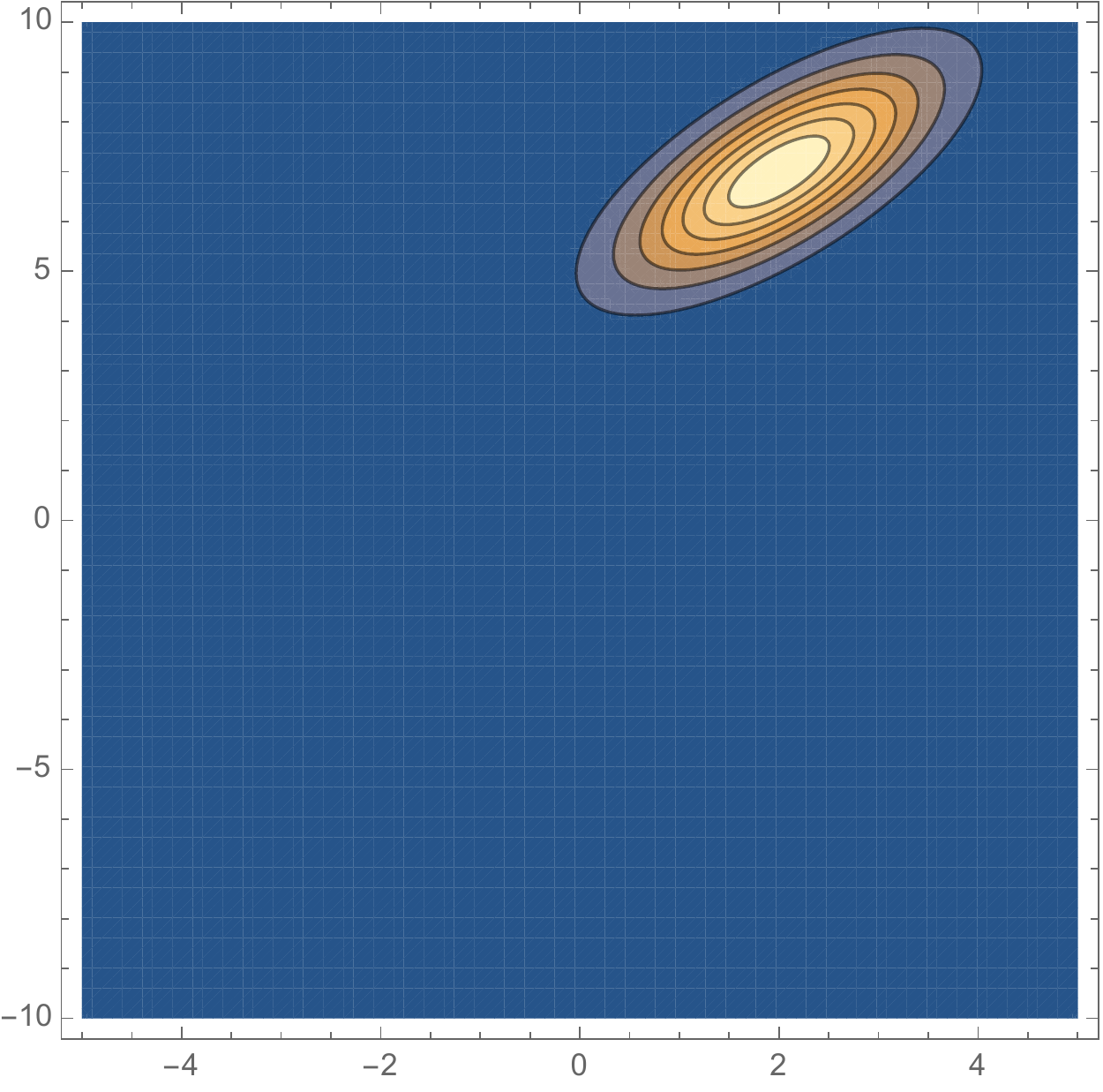} 
\caption{The Wigner function $W(\Sigma_1, \mu_1)$.}
\end{minipage}
\label{Figurecrit}
\end{figure}

Recall that our objective is to find geodesics on the parameter manifold $(\Theta, g_{\theta})$ that minimize the square geodesic distance
 \begin{equation}
W^2_Q(\rho(\bullet,\theta^0),\rho(\bullet,\theta^1)) = \inf_{\substack{\theta \in C^1(0,1)\cap C[0,1]\\ \theta(0)=\theta^0 , \theta(1)=\theta^1}}\left\{ \int_0^1 \langle \dot{\theta}(t), G_W(\theta(t)) \dot{\theta}(t)\rangle \ dt \right\}.
 \end{equation}

We then discretize the integral of the optimal control problem as

\begin{equation}
\begin{split}
&\min_{\theta_i; 1\le i \le N-1} N^{-1} \sum_{i=1}^{N-1} \left\langle \left(\frac{\theta_{i+1}-\theta_i}{N} \right), G_W(\theta_i) \left( \frac{\theta_{i+1}-\theta_i}{N} \right) \right\rangle \\
&=\min N^{-3} \sum_{i=1}^{N-1} \left( \left\lVert \mu_{i+1}-\mu_{i}\right\rVert^2 + \tr\left( (S_{\theta_{i+1}}-S_{\theta_{i}}) \Sigma_i (S_{\theta_{i+1}}-S_{\theta_{i}}) \right) \right)
\end{split}
\end{equation}
with boundary conditions $\theta_0 = \theta^0$ and $\theta_N=\theta^1.$

This minimization problem can be easily solved using a simple Monte-Carlo algorithm minimizing \eqref{eq:functional} that only accepts transitions to states that satisfy the two constraints
\[ \Sigma_i\ge 0 \text{ and } \Sigma_i + i \nu \ge 0.\]
The numerical solution to the quantum transport problem of the two parametrized Gaussian states is illustrate in Figure \ref{fig:optimaltransport}.
\end{ex}

\begin{figure}
\includegraphics[width=\linewidth]{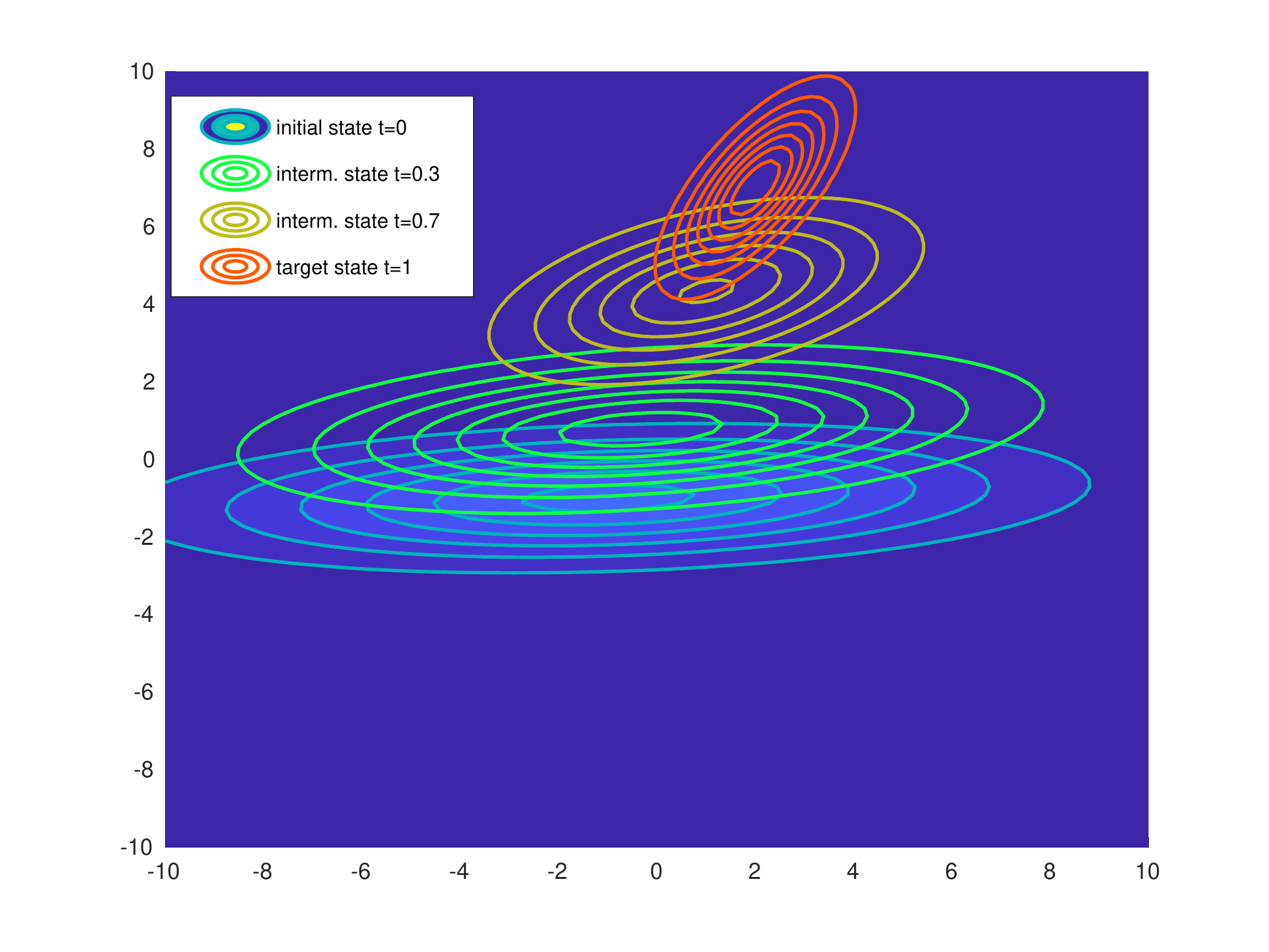}
\caption{Optimal quantum transport map from quantum state with Wigner function $W(\theta^0)$ to quantum state with Wigner function $W(\theta^1)$.}
\label{fig:optimaltransport}
\end{figure}

\subsection{Examples involving the quantum fermionic Fokker-Planck equation}
\label{sec:QFFPE}
\begin{ex}[Fermionic Fokker-Planck equation; Analytic solution]
\label{ex:1}
We consider the fermionic Fokker-Planck equation as introduced in Subsection \ref{subsec:FFP} for simplest case $n=1,$ i.e. $\mathfrak C$ can be identified with the two-dimensional Hilbert space $\operatorname{span}\{ \operatorname{id}_{\CC^{2 \times 2}},\sigma_x\}$ in which case we can solve the problem analytically.
 
 The grading operator is defined by 
\[\Gamma(\operatorname{id})=\operatorname{id}\text{ and }\Gamma(\sigma_x)=-\sigma_x.\]

The faithful states in $\mathfrak C$ are then parametrized by 

\[ (-1,1) \ni \theta \mapsto \rho(\theta):= \operatorname{id}+\theta \sigma_x.\]
We can diagonalize this density operator using the unitary map $U= 2^{-1/2}(\sigma_x- \sigma_z).$ This way, 
$U \rho(\theta) U = \operatorname{diag}(1+\theta,1-\theta).$
The derivative is given by 
\[ \nabla (\alpha \operatorname{id}+\beta \sigma_x) = \beta  \operatorname{id}.\]
The operator $L_{\rho(\theta),\Gamma(\rho(\theta))} (\operatorname{id})= \int_0^1 (\rho(\theta))^{1-s}(\rho(-\theta))^{s} \ ds$ becomes therefore after conjugating with $U$

\[UL_{\rho(\theta),\Gamma(\rho(\theta))} (\operatorname{id})U=\int_0^1 (1-\theta)^{1-s} (1+\theta)^{s} \ ds \operatorname{id}_{\CC^{2 \times 2}} = \frac{\theta}{\operatorname{artanh}(\theta)}\operatorname{id}_{\CC^{2 \times 2}}.\]

This implies that $-\Delta_{\rho(\theta)} \vert_{\operatorname{span}(\sigma_x)}=\frac{\theta}{\operatorname{artanh}(\theta)} \operatorname{id}.$ Using that $D_{\theta}\rho(\theta)=\sigma_x$ and that $G_{\theta}=\operatorname{id}$, we find from \eqref{eq:GW} that
\[G_W(\theta)= \frac{\operatorname{artanh}(\theta)}{\theta}.\]

As before, our objective is to find geodesics on the parameter manifold $(\Theta, g_{\theta})$ that minimize the square geodesic distance
 \begin{equation}
W^2_Q(\rho(\bullet,\theta^0),\rho(\bullet,\theta^1)) = \inf_{\substack{\theta \in C^1(0,1)\cap C[0,1]\\ \theta(0)=\theta^0 , \theta(1)=\theta^1}}\left\{ \int_0^1 \mathcal L(\theta(t),\dot{\theta}(t)) \ dt \right\}
 \end{equation}
where $\mathcal L(\theta(t),\dot{\theta}(t)):=\dot{\theta}(t)^2 G_W(\theta(t))$ is the Lagrangian. The associated Euler-Lagrange equation
\[ \partial_1 \mathcal L(\theta(t),\dot{\theta}(t)) - \frac{d}{dt} \partial_2 \mathcal L(\theta(t),\dot{\theta}(t))=0\]
 becomes
\begin{equation}
\begin{split}
 &\dot{\theta}(t)^2 G_W'(\theta(t))- 2 \frac{d}{dt} (\dot{\theta}(t)G_W(\theta(t)))\\
 &=-\dot{\theta}(t)^2 G_W'(\theta(t))-2\ddot{\theta}(t) G_W(\theta(t))=0.
 \end{split}
\end{equation}
Using that $G_W(\theta)>0$, we find the identities for $\pm \dot{\theta}(t) > 0$
\[  \frac{d}{dt}\log(G_W(\theta(t)))=\frac{G_W'(\theta(t))\dot{\theta}(t)}{G_W(\theta(t))} \text{ and } \frac{d}{dt} \log(\pm \dot{\theta}(t))= \frac{\ddot{\theta}(t)}{\dot{\theta}(t)}.\]
Assuming that $\theta^1>\theta^0$ in the sequel and thus dropping $\pm$ for simplicity, the Euler-Lagrange equation is then equivalent, for some constant $k \in \RR$, to the ODE
\[ -\log(G_W(\theta(t)))+k=  \log(\pm \dot{\theta}(t)).\]
Introducing then the function $\zeta(x):=\frac{\operatorname{Li}_2(x)
-\operatorname{Li}_2(-x)
}{2}$ in terms of the dilogarithm, $\operatorname{Li}_2$, we can then specify the constant $e^k$ by
\[ e^k:=\int_0^1  \frac{\dot{\theta}(t)}{\theta(t)} \operatorname{artanh}(\theta(t)) \ dt = \int_{\theta^0}^{\theta^1} \frac{\operatorname{artanh}(x)}{x} \ dx =\zeta(\theta^1)-\zeta(\theta^0).\]
In particular, this allows us to explicitly state the solution to the optimal transport problem
\begin{equation}
\begin{split}
\label{eq:solution}
 \theta(t)= \zeta^{-1} \left( t \zeta(\theta^1)+(1-t) \zeta(\theta^0) \right).
\end{split}
\end{equation}
We illustrate this by choosing three different pairs of parameters
\begin{equation}
\begin{split}
\label{eq:three}
\theta^0_1&=-\tfrac{1}{2}, \ \theta^1_1=\tfrac{1}{2}, \\ 
\theta^0_2&=-\tfrac{9}{10}, \ \theta^1_2=\tfrac{9}{10}, \\ 
\theta^0_3&=-\tfrac{999}{1000}, \ \theta^1_3=\tfrac{999}{1000}.
\end{split}
\end{equation}

\begin{figure}
\includegraphics[width=12cm]{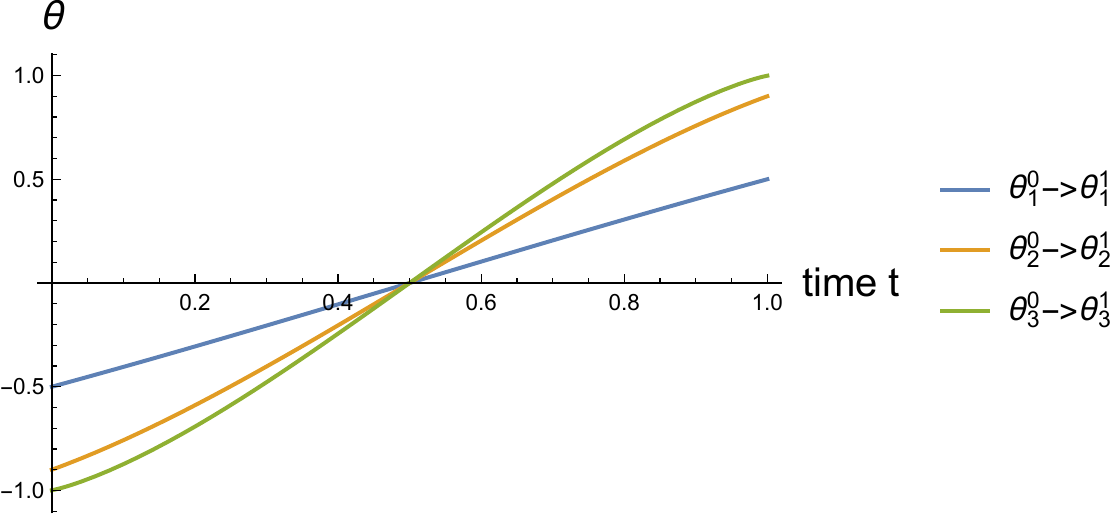}
\caption{Illustration of solution \eqref{eq:solution} for three different boundary parameters \eqref{eq:three}.}
\end{figure}

\subsubsection{Anti-commutator case}

We can repeat the previous analysis by considering instead of the Feynman-Kubo-Mori multiplication operator $L_{\rho(\theta),\Gamma(\rho(\theta))}$ the anti-commutator \eqref{eq:anti-comm} which satisfies
\begin{equation}
L^{\text{ac}}_{\rho(\theta)}(\operatorname{id})=\rho(\theta).
\end{equation}
Thus, using that $\nabla^*(\rho(\theta))=1$, we find that $-\Delta_{\rho(\theta)} \vert_{\operatorname{span}(\sigma_x)} =1$ and therefore also 
\[ G_W(\theta)\equiv 1.\]
In particular, since the Lagrangian is just $\mathcal L(\dot{\theta}(t)) = \dot{\theta}(t)^2$, the geodesics in parameter space are just straight lines as the Euler-Lagrange equation $\ddot{\theta}(t)=0$ immediately shows.
\end{ex}

\subsubsection{Wasserstein natural gradient}
\label{sec:WNG}
We shall now also illustrate the Wasserstein natural gradient for the quantum Fokker-Planck equation as in Example \ref{ex:1} by minimizing the von Neumann entropy as objective function 
\[ R(\theta) = \tau(\rho(\theta) \log(\rho(\theta))).\]
From the matrix logarithm
\begin{equation}
 \log(\rho(\theta)) =\tfrac{1}{2} \begin{pmatrix} \log(1-\theta^2) & \log\left(\frac{1+\theta}{1-\theta} \right) \\ \log\left(\frac{1+\theta}{1-\theta} \right) & \log(1-\theta^2) \end{pmatrix}
 \end{equation}
 we then immediately see that 
\[R(\theta)= \frac{1}{2} \left( \log(1-\theta^2)  + \theta \log\left(\frac{1+\theta}{1-\theta} \right)\right)\] 
and hence $D_{\theta} R(\theta)= \operatorname{artanh}(\theta).$
Therefore, the Wasserstein gradient \eqref{eq:WGN} becomes $\nabla_g R(\theta) = G_W(\theta)^{-1} D_{\theta} R(\theta) = -\theta.$
The gradient descent equation therefore becomes in parameter space 
\[ \theta'(t) = -\theta(t)\]
which implies that we will converge exponentially fast to the unique minimizer the completely mixed state that corresponds to $\theta=0.$

\subsection{Channel parameter estimation-pushforward of quantum states}
\label{sec:CPE}
The idea of parameter estimation of probability densities constructed from the pushforward of possibly nonlinear activation functions, relevant for neural networks, has been investigated by the second author in \cite{LZ20}. 

In quantum theory the framework is somewhat different, since quantum operations on a physical system are described by linear (super)-operators, so-called quantum channels rather than non-linear one-dimensional functions. A quantum channel is a completely positive and trace preserving (CPTP) map. 
Thus, it is natural to consider the situation where a state is parametrized by the output of a quantum channel $\Phi_{\theta}$ depending on some parameter $\theta$ which is the quantum analogue of the pushforward of probability measures by parametrized functions. 

We shall  illustrate how such problems can be studied in our framework by considering the quantum depolarizing channel \eqref{eq:depchann} with the quantum fermionic Fokker-Planck equation, introduced in Section \ref{subsec:FFP}, for $n=2$.

\begin{ex}[Depolarizing channel and quantum Fokker-Planck dynamics]
Consider the fermionic Fokker-Planck equation with $n=2$, the fermionic operators are then given by
\[ Q_1:=\sigma_x \otimes \operatorname{id}_{\CC^2} \text{ and } Q_2:=\sigma_z \otimes\sigma_x \]
and thus 
\begin{equation}
\begin{split}
Q^{(0,0)}&=  \operatorname{id}_{\CC^4}, \ 
Q^{(1,0)}=  \sigma_x \otimes \operatorname{id}_{\CC^2}, \
Q^{(0,1)}=  \sigma_z \otimes\sigma_x, \ Q^{(1,1)}= -i \sigma_y \otimes \sigma_x.
\end{split}
\end{equation}
Thus, we find for the gradients
\begin{equation}
\begin{split}
\nabla (Q^{0,0}) &= \mathbf 0, \quad 
\nabla (Q^{1,0}) = (Q^{(0,0)}, 0)\\
\nabla (Q^{0,1}) &= (0, Q^{(0,0)}), \text{ and } 
\nabla (Q^{1,1}) = (Q^{(0,1)},-Q^{(1,0)}).  
\end{split}
\end{equation}
We consider the depolarizing channel for some density operator $\rho =\frac{1}{2}( Q^{(1,0)}+Q^{(0,0)})$ and limiting state $\frac{1}{2}( Q^{(0,1)}+Q^{(0,0)})$ 
\begin{equation}
\label{eq:depchann}
\Phi_{\theta}(\rho) =\frac{1}{2}\left( e^{-\theta}Q^{(1,0)} + (1-e^{-\theta}) Q^{(0,1)}+Q^{(0,0)}\right).
\end{equation}
Then, after applying the anti-commutation operator \eqref{eq:anti-comm}
\begin{equation}
\begin{split}
L^{\text{ac}}_{\Phi_{\theta}(\rho)}(\nabla (Q^{1,0}))&= \frac{1}{2}(e^{-\theta}Q^{(1,0)} +(1-e^{-\theta})Q^{(0,1)}+Q^{(0,0)}, 0) \\
L^{\text{ac}}_{\Phi_{\theta}(\rho)}(\nabla (Q^{0,1}))&=  \frac{1}{2}(0,e^{-\theta} Q^{(1,0)} + (1-e^{-\theta}) Q^{(0,1)}+Q^{(0,0)}) \\
L^{\text{ac}}_{\Phi_{\theta}(\rho)}(\nabla (Q^{1,1}))&=\frac{1}{2}((1-e^{-\theta})Q^{(0,0)}+Q^{(0,1)},-e^{-\theta}Q^{(0,0)}-Q^{(1,0)})
\end{split}
\end{equation}
we find for the Laplacian 
\begin{equation}
\begin{split}
-\Delta_{\Phi_{\theta}}(Q^{(1,0)})&=\frac{1}{2}\left( (1-e^{-\theta})Q^{(1,1)}+Q^{(1,0)} \right)\\
-\Delta_{\Phi_{\theta}}(Q^{(0,1)})&=\frac{1}{2}\left( -e^{-\theta}Q^{(1,1)}+Q^{(0,1)}\right) \\
-\Delta_{\Phi_{\theta}}(Q^{(1,1)})&= \frac{1}{2}\left((1-e^{-\theta})Q^{(1,0)}- e^{-\theta}Q^{(0,1)}\right)+Q^{(1,1)}.
\end{split}
\end{equation}

This means that the Laplacian has a matrix representation 
\begin{equation}
\begin{split}
-\Delta_{\Phi_{\theta}}\vert_{\operatorname{span}\{\operatorname{id}\}^{\perp}} =\frac{1}{2} \begin{pmatrix}1 & 0 & 1 - e^{-\theta}   \\ 0&1&  - e^{-\theta} \\
1 - e^{-\theta}  & - e^{-\theta}  & 2 \end{pmatrix} 
\end{split}
\end{equation}
with inverse
\begin{equation}
\begin{split}
(-\Delta_{\Phi_{\theta}}\vert_{\operatorname{span}\{\operatorname{id}\}^{\perp}})^{-1} =\frac{2}{2 e^\theta+e^{2 \theta}-2}  \left(
\begin{array}{ccc}
 2 e^{2 \theta}-1& 1-e^{\theta} & e^\theta
   \left(1-e^\theta \right) \\
 1-e^\theta & 2 e^\theta+e^{2 \theta}-1 & e^\theta \\
 e^\theta
   \left(1-e^\theta \right) & e^\theta  &
   e^{2 \theta} \\
\end{array}
\right).
\end{split}
\end{equation}
Since, $D_{\theta} \Phi_{\theta}(\rho)= \frac{e^{-\theta}}{2}\left(  Q^{(0,1)}-Q^{(1,0)} \right)$ we thus find 
\begin{equation}
\begin{split}
G_W(\theta)&=\tau( D_{\theta} \Phi_{\theta}(\rho) (-\Delta_{\Phi_{\theta}}\vert_{\operatorname{span}\{\operatorname{id}\}^{\perp}})^{-1} D_{\theta} \Phi_{\theta}(\rho)) \\
&=\frac{e^{-2\theta}\left((2e^{2\theta}-1)-2(1-e^{\theta})+(2e^{\theta}+e^{2\theta}-1) \right)}{4e^{\theta}+2e^{2\theta}-2}= \frac{3-4e^{-2\theta}+6e^{-\theta}}{4e^{\theta}+2(e^{2\theta}-1)}.
\end{split}
\end{equation}
\end{ex}

\section{Discussion}
In this paper, we pull back the quantum Wasserstein-2 metric into a parameterized quantum statistical models. This allows us to develop a quantum Wasserstein/transport information matrix. Using this matrix, we develop the quantum transport natural gradient methods and apply them to the quantum statistical learning problems. Besides, we also consider the optimal control problem of quantum transport natural gradient flows, which leads to the derivation of quantum Schr{\"o}dinger bridge problem. Several analytical examples, such as the transport of Gaussian states on the statistical manifold in Example \ref{ex:Wigner1}, the transport of states for the gradient induced by quantum fermionic Fokker-Planck equation in Section \ref{sec:QFFPE} on the statistical manifold, and the parameter estimation problem for channels in Subsection \ref{sec:CPE}, are provided. 

Our results initialize the joint study among quantum information geometry and quantum optimal transport. We pull back the quantum system dynamics into a finite-dimensional parameter space generated by statistical and machine learning models. We call this area {\em quantum transport information geometry}. Here the interaction study between quantum Fisher and quantum Wasserstein information matrices becomes essential. We expect that this joint study would be useful in developing transport estimation theory of quantum information theory, and designing AI-driven quantum computing algorithms for quantum systems. In the future, we will continue this line of study following transport information geometry \cite{Li1, Li2}.  

\smallsection{Acknowledgements} 
This work was supported by the EPSRC grant EP/L016516/1 for the University of Cambridge CDT, the CCA and generous funding from IPAM in Los Angeles (S.B.). 
Wuchen Li was supported by University of South Carolina, start up funding. 
%


\end{document}